\documentclass[letterpaper, 10 pt, conference]{ieeeconf}  
\IEEEoverridecommandlockouts                              

\usepackage{verbatim}
\usepackage{amsmath}
\usepackage{amssymb}

\usepackage[noadjust]{cite}

\usepackage{wasysym}

\makeatletter
\newtheorem{thm}{\protect\theoremname}
\newtheorem{lem}[thm]{\protect\lemmaname}
\newtheorem{rem}[thm]{\protect\remarkname}
\newtheorem{prop}[thm]{\protect\propositionname}

\def\rea{\mathbb{R}}
\def\col{{\rm col}}
\def\calm{\mathcal{M}}
\def\frm{\mathfrak{M}}
\def\he{{\rm He}}
\def\calx{\mathcal{X}}
\usepackage{color}

\definecolor{darkgreen}{rgb}{0,0.6,0}
\usepackage{hyperref}
\hypersetup{
    colorlinks=true,
    linkcolor=darkgreen,
    filecolor=magenta,      
    urlcolor=cyan,
    }

\usepackage[caption=false,font=footnotesize]{subfig}
\DeclareMathOperator*{\Tr}{Tr}

\makeatother

\providecommand{\lemmaname}{Lemma}
\providecommand{\propositionname}{Proposition}
\providecommand{\remarkname}{Remark}
\providecommand{\theoremname}{Theorem}

\title{Control Contraction Metrics on Lie Groups }

\author{Dongjun Wu$^{1}$, Bowen Yi$^{2}$, Ian R. Manchester$^{3}$ 
	\thanks{*This project has received funding from the European Research Council (ERC) under the European Union's Horizon 2020 research and innovation programme under grant agreement No 834142 (ScalableControl).}
	\thanks{$^{1}$ D. Wu is with Department of Automatic Control, Lund
	University, Box 118, SE-221 00 Lund, Sweden {\tt\small
	dongjun.wu@control.lth.se}.}
	\thanks{$^{2}$ B. Yi is with Department of Electrical Engineering, Polytechnique Montreal, QC H3T 1J4, Canada {\tt \small
bowen.yi@polymtl.ca}.}
	\thanks{$^{3}$ I. R. Manchester is with Australian Centre for Robotics and School of Aerospace, Mechanical and Mechatronic Engineering, The University of Sydney, Australia 
 {\tt \small ian.manchester@sydney.edu.au}.}}
\begin{document}
\maketitle

\thispagestyle{empty}
\pagestyle{empty}

\begin{abstract}
In this paper, we extend the control contraction metrics (CCM) approach, which was originally proposed for the universal tracking control of nonlinear systems, to those that evolves on Lie groups. Our idea is to view the manifold as a constrained set that is embedded in Euclidean space, and then propose the sufficient conditions for the existence of a CCM and the associated controller design. Notably, we demonstrate that the search for CCM on Lie groups can be reformulated as convex conditions. The results extend the applicability of the CCM approach and provide a framework for analyzing the behavior of control systems with Lie group structures.

\end{abstract}

{\keywords nonlinear systems, contraction analysis, Lie groups, tracking control}

\section{Introduction}

Contraction analysis has gained widespread recognition both within and beyond the control community since the seminal paper by Lohmiller and Slotine  \cite{lohmiller1998contraction}. It provides a powerful and flexible tool to analyze the stability and dynamical behaviour of nonlinear systems by means of linear systems theory. Over the years, researchers have extensively applied contraction analysis to various domains, including
observer design \cite{Yi2021,aghannan2003intrinsic}, trajectory tracking
\cite{reyes2017tracking,pavlov2008incremental}, machine learning
\cite{Blocher2017,Singh2021,revay2023recurrent,yi2023equivalence}, and motion planning \cite{Singh2019,Tsukamoto2021a}. For a comprehensive review of the literature and future perspectives on this field, the interested reader may refer to the recent survey paper \cite{Tsukamoto2021} and the monograph \cite{Bullo2023book}.

One recent development in contraction analysis involves its utilization as a synthesis tool for constructive nonlinear control, exemplified by the widely popular \emph{control contraction metric} (CCM) method introduced in \cite{manchester2017control}. It is shown that by searching for a CCM, the system could be \emph{universally} stabilized for any feasible reference trajectories. A salient feature of this approach is that its design procedure
can be turned into a convex optimization problem. See also \cite{manchester2018robust} for its robust version for non-affine systems. In the last few years, the CCM approach has been successfully applied to many domains, such as the safe control of robots \cite{dawson2023safe}, adaptive control \cite{lopez2020robust}, model predictive control (MPC) \cite{sasfi2023robust}, and motion planning \cite{Tsukamoto2021a,Singh2019}.

The existing results on CCM primarily focus on systems evolving on Euclidean space. However, a majority of nonlinear systems reside on manifolds, including
robotic models and rigid body models of vehicles such as aerial or underwater robots.  It is still an open problem to extend the CCM approach to manifolds. In this paper, we propose some preliminary results for such an extension to Lie groups. The main contributions of the paper are twofold:

\begin{itemize}
\item[1.] We formulate the CCM on embedded submanifolds and then specialize it to Lie groups setting. Subsequently, we demonstrate that the search for a CCM on Lie groups can be characterized by some convex conditions.

\vspace{.4em}

\item[2.] In the sampled-data version of the ``standard'' CCM, it is required to compute the geodesic during the online implementation at every sampling time. We show that computation of geodesics can be avoided by using other types of curves. It provides easier solutions while ensuring guaranteed stability.
\end{itemize}

{\it Notations:}  Given a manifold $\mathcal{M}$, $T \mathcal{M}$ and $T^*
\mathcal{M}$ represent the tangent and co-tangent bundles of $\mathcal{M}$. We use $g(\cdot, \cdot)$ and $d$ to denote the Riemannian metric and the
associated Riemannian distance, respectively. $\nabla$ stands for the Levi-Civita connection.
$\Tr(\cdot)$ is the trace of a square matrix, and $\|\cdot\|$ is the Euclidean
$2$-norm. For symmetric matrices $A, B \in \mathbb{R}^{n\times n}$, $A>B$ means that the matrix $(A-B)$ is positive definite; for a square matrix $X$, we define ${\he}\{X\}:= X+ X^\top$. Given a vector field $f$, $D_f Q$ stands for the directional
derivative of a smooth quantity $Q$ along $f$, 
i.e. $D_f Q  = \sum_j \frac{\partial Q}{\partial x_j} f_j$, particularly $L_f Q$  representing the Lie derivative. For $m \in \mathbb{N}_+$, denote the set $\ell_m := \{ 1, \cdots, m \}$.

\section{Preliminaries}

In this section, we review some preliminary results about the control contraction metric (CCM) approach \cite{manchester2017control}.

Consider the nonlinear system
\begin{equation}
\label{sys:nltv}
\dot{x}=f(x,t)+B(x,t)u
\end{equation}
with the state $x\in \calx \subseteq \rea^n$ and the input $u\in\mathbb{R}^{m}$, 
where the functions $f: \rea^n \to \rea^n$ and $B(x)=\col (b_{1}(x),\cdots,b_{m}(x))$
are assumed continuously differentiable, and $B$ has constant rank $m$.
\footnote{The constant rank condition of $B(x,t)$ is related to the
universal exponential stabilizability of the system \eqref{sys:nltv}.
}
The paper \cite{manchester2017control} considers the Euclidean case with $\calx = \rea^n$. We call $(x_\star (t),u_\star(t)) \in \calx \times \rea^m$ a feasible pair, if it satisfies $\dot{x}_{\star}=f(x_{\star},t)+B(x_{\star},t)u_{\star}$ for all
$t\ge0$. 
\vspace{.4em}
\\ \
\textbf{Problem set.} For the given system \eqref{sys:nltv}, we aim to find a feedback law $u = k_p (x, x_\star, u_\star, t)$ that exponentially stabilizes any feasible desired trajectory $x_\star$. 

The first step in contraction analysis
is to calculate the \emph{variational dynamics} of the system \eqref{sys:nltv}, that is
\begin{equation}
\label{sys:var}
\delta\dot x=A(x,u,t)\delta x+B(x,t)\delta u
\end{equation}
with $A(x,u,t)=\frac{\partial f}{\partial x}(x,t) + \sum_{i=1}^{m}\frac{\partial b_{i}}{\partial x}u_{i}$. The variables $\delta x \in T\calx$ and $\delta u \in T\rea^m$ are the differential state and input, which are the infinitesimal variations in terms of the original system \eqref{sys:nltv}; see \cite{manchester2017control,VAN13} for details. A notable feature is that the variational system \eqref{sys:var} is linear time-varying (LTV) if we regard $x$ and $u$ as exogenous signals.

In the CCM approach, the design of a tracking controller is reformulated as the search for a smooth control contraction metric $M: \rea^{n} \times \rea \to \rea^{n\times n}_{>0}$
and a differential controller $\delta u = k_\delta (x,\delta x, u, t)$ satisfying 

\begin{itemize}
    \item[\bf C0] The metric $M$ is uniformly bounded, i.e.
\begin{equation} 
\alpha_{1}I\le M(x,t)\le\alpha_{2}I,\ \ \forall x,t\label{CCM-1}
\end{equation}
for $\alpha_{1},\alpha_{2}>0$ and satisfies the contraction condition
\begin{equation}
\label{CCM-2}
\begin{aligned}
B^\top &M \delta x  = 0  \implies
\\
&\delta x^{\top}
\left( 
\frac{\partial M}{\partial t} + D_f M
+ {\rm He}\{ M A\}
+ 2\lambda M 
\right) \delta x <0 
\end{aligned}
\end{equation}
for all $x, u , t$ and non-zero $\delta x $,
with some $\lambda>0$.
\end{itemize}

As proven in \cite{manchester2017control}, the following two stronger
conditions imply \eqref{CCM-2}:

\begin{itemize}
    \item[\bf C1] For $\delta x \ne 0$, the implication holds
\end{itemize}
\begin{equation}
\begin{aligned}
B^\top &M  \delta x = 0  \implies
\\
& \delta x^{\top}
\left( \frac{\partial M}{\partial t}
+ D_f M 
+ \he\left\{\frac{\partial f}{\partial x}^\top M \right\}
+ 2\lambda M 
\right) \delta x <0.
\end{aligned}
\end{equation}
\begin{itemize}
    \item[\bf C2] For each $i\in \ell_m$,
\end{itemize}
\begin{equation} \label{Killing}
    D_{b_i} M + 
    \he\left\{\frac{\partial b_i}{\partial x}^\top M \right\} 
    =0. 
\end{equation}

When \eqref{Killing} is satisfied, we call $b_i$ a Killing vector field
under the metric $M(x,t)$, and in this case,
$M(x,t)$ is referred to as a {\it strong} CCM. Otherwise, if only {\bf C1} holds, we call it a {\it weak} CCM.

\begin{rem}
It should be noted that the (global) existence of Killing fields 
on a {\it manifold} is related to the topological properties of the manifold \cite{nomizu1960local}, making a relatively strict condition. It might be easier to find a weak CCM for a given system on manifolds. However, the challenge can be circumvented if we consider locally.
\end{rem}


Once a CCM has been found, a controller can be synthesized in the
following two steps.
\begin{itemize}
    \item[\bf S1] Find a differential feedback 
    $\delta u = k_{\delta} (x, \delta x, u, t)$ such that 
    \begin{equation*}
        \delta x^\top \dot{M} \delta x 
        + 2\delta x^\top M (A\delta x + Bk_\delta) 
        < -2\lambda \delta x^\top M \delta x
    \end{equation*}
    holds for all $x, u, t$ and non-zero $\delta x$.
    
    \item[\bf S2] Design a feedback law $u = k_p(\cdot)$ to ensure that the closed-loop dynamics conform to the desired differential systems when subjected to the differential feedback $\delta u$ that is obtained in {\bf S1}.
\end{itemize}

For {\bf S2}, a feasible construction is that, given a minimizing geodesic $\gamma:[0,1]\to\mathbb{R}^{n}$
    joining $x_\star (t_{0})$ to $x(t_{0})$, the feedback $k_{p}$ can be selected as the 
    solution to
    \[\small
    \begin{aligned}
    k_{p} (c, u_\star , t, s)  
    =  u_{\star}(t) + 
    \int_{0}^{s}k_{\delta}\left(c(\mathfrak{s}),{c'(\mathfrak{s})},k_{p}(c,u_{*},t,\mathfrak{s}),t \right)d\mathfrak{s}
    \end{aligned}
    \]
    where $c(t,s)$ is the solution to the system \eqref{sys:nltv} with initial
    condition $\gamma(s)$. Then $u=k_{p}(c,u_{\star},t,1)$ exponentially stabilizes the trajectory
$x_{\star}(\cdot)$, i.e., 
$
d(x_{\star}(t),x(t))\le e^{-\lambda t}d(x_{\star}(0),x(0)),\ \ \forall t\ge0.
$

The control strategy described above is open-loop control as there is only onetime measurement of the state. Its performance and robustness can be enhanced by considering the sampled-date controller -- recursively applying the open-loop control described above at sampled instances \cite{manchester2017control} -- or by adopting the minimal geodesics in real time \cite{wang2020continuous}. 

\section{Control Contraction Metrics on Lie Groups}

\subsection{General results}

The Lie groups commonly encountered in control systems have a structure characterized as matrix Lie groups multiplied by (in direct product with) a vector space. Many of them, if not all, can be seen as ``constrained
submanifolds'' within certain Euclidean spaces, e.g., $SO(n)=\{X\in\mathbb{R}^{n\times n}:X^{\top}X=I,~\det X=1\}$.
Therefore, it is natural to view these as systems in Euclidean space
with constraints 
\begin{equation}
\label{eq:M}
    \mathcal{M}=\{x\in\mathbb{R}^{n}:h(x)=0\},
\end{equation}
where the smooth function $h: \rea^n \to \rea^q$ has constant rank. 

In this section, we consider the system model in the form \eqref{sys:nltv} with the state space $\calx $ being $\mathcal{M}$. To guarantee that the vector fields $f$ and $b_i$ ($ i \in \ell_m$) are tangent
to the manifold, it is necessary and sufficient to satisfy the transversality condition
\begin{equation}
\label{eq:tanspace}
L_f h(x) = 0, \quad L_{b_i} h(x) =0.
\end{equation} 
for all $x\in\mathcal{M}$ and $t\ge0$. We aim to extend the CCM approach to the systems on manifolds to exponentially track a feasible desired trajectory $x_\star (t) \in \calm$.

Mimicking the conditions {\bf C0}-{\bf C2}, we propose the corresponding version on the manifold $\calm$ as follows. A natural idea is to equip the manifold with the induced metric within $ \calm \times \rea_{\ge 0}$ that is required to be positive definite on the tangent bundle $T \calm$. These conditions are sufficient to design a CCM controller on manifolds, which will be introduced in Proposition \ref{prop:1} below.

\begin{itemize}
\item[\bf A0] There exist two positive constants $a_{1},a_{2}$ such
that
\end{itemize}
\begin{equation*}
\frac{\partial h}{\partial x}\delta x=0 \implies  a_{1}\|\delta x\|^{2}\le\delta x^{\top}M(x,t)\delta x\le a_{2}\|\delta x\|^{2}\label{A0}
\end{equation*}

\begin{itemize}
\item[\bf A1] For some $\lambda>0$ and $\forall \delta x\ne0$, the following holds
\end{itemize}
\begin{align} \label{A1}&
\left. \begin{aligned}
B^\top M \delta x & =0\\
\dfrac{\partial h}{\partial x}\delta x &=0
\end{aligned}\right\}
\implies
\\
 & \delta x^{\top}\left(\frac{\partial M}{\partial t}  
  +\partial_{f}M 
  + \he \left\{\frac{\partial f}{\partial x}^\top M \right\}
  +2\lambda M\right)\delta x<0 \nonumber
\end{align}

\begin{itemize}
\item[\bf A2] Each $b_i$ is a Killing field, i.e.,
for all $\delta x_1, \delta x_2$
\end{itemize}
\begin{equation}
\begin{aligned}
&{\partial h \over \partial x}(x)  \delta x =0
    \implies \\ 
    & \hspace{4em}
    \delta x^{\top}\left(D_{b_{i}}M
    + \he \left\{M\frac{\partial b_{i}}{\partial x} \right\}
    \right)\delta x=0.\label{A2}
\end{aligned}
\end{equation}


\begin{rem}
In contrast to the Euclidean case, the new ingredient in {\bf A0}-{\bf A1} is 
$\frac{\partial h}{\partial x} \delta x = 0$,
which imposes the constraint $\delta x \in T \calm$. Similarly to {\bf C2}, the condition {\bf A2} is related to the uniform stabilizability of the differential dynamics for arbitrary $u_\star$ \cite{van2013differential,manchester2017control}. However, our case only requires the condition on the set $\{x \in \rea^n : h(x) =0\}$.\footnote{In fact, the condition \eqref{A2} coincides with the definition of Killing vector fields i.e. $\mathcal{L}_b M(x,t) =0$, which should not be surprising. Indeed, 
$
   \mathcal{L}_b M (\delta x, \delta x)= 
    \left. \frac{d}{dt} \right|_{t=0} 
    \left(\frac{\partial \phi^b_t(x)}{\partial x} 
    \delta x \right)^\top 
    M(t,\phi^b_t(x)) 
    \frac{\partial \phi^b_t(x)}{\partial x} 
       \delta x 
$
where $\phi^b_t: x \mapsto \phi^b_t (x)$ is the flow of the vector field $b$. 
Expanding the above equation results in \eqref{A2}. }
\end{rem}

Once a function $M(x,t)$ satisfying {\bf A0}-{\bf A2} can be found, we call it a CCM on $\calm$. 
Then, we can proceed to controller design as follows. Assume that $(x_{\star}(\cdot),u_{\star}(\cdot))$ is a feasible pair of the system
\eqref{sys:nltv}
and $\gamma:[0,1]\to\mathcal{M}$ is the minimizing geodesic under the
induced metric from $\mathbb{R}^{n}$
joining $x_{\star}(t_{0})$ to $x(t_{0})$. We have the following.

\begin{prop}
\label{prop:1}
Let $M(x,t)$ be a CCM on $\calm$ satisfying {\bf A0 - A3}.
There exist open-loop and sampled-data exponentially stabilizing controllers 
of the trajectory $x_\star(\cdot)$.

\end{prop}
\begin{proof}
In the proof, we show the existence of the open-loop and sampled-data controllers that exponentially stabilize the tracking error system.

1) \emph{Open-loop controller}.
Let $u(t,s), x(t,s)$ be the solution to the partial differential system
\begin{equation}\label{eq:PDE}
\begin{aligned}
\dfrac{\partial x}{\partial t} &= 
f(x,t) + B(x,t)u \\
\dfrac{\partial u}{\partial s} &=
- \dfrac{1}{2} \rho(x,t)
B(x,t)^\top {M} (x,t)\dfrac{\partial x}{\partial s} 
\end{aligned}
\end{equation}
with $s\in[0,1]$, $u(t,0) = u_\star (t),~ \forall t \ge t_0$ and $x(t_{0},s)=\gamma(s),~ \forall s\in[0,1]$, where the scalar function $\rho(x,t) \ge 0$ will be determined later in the proof to guarantee the contraction of the closed loop.

Considering the infinitesimal variable $\delta x(t) := {\partial x \over \partial s}(t,s)$, we calculate the evolution of energy:
\begin{align*}
  & \frac{d}{dt}\int_0^1 \delta x^\top M \delta x ds \\
  = & \int_0^1 2\delta x^\top  M
  \left( \frac{\partial f}{\partial x}\delta x + \sum_{i=1}^m u_i 
  \frac{\partial b_i}{\partial x} \delta x + \sum_{i=1}^m 
  \frac{\partial u_i}{\partial s}b_i
  \right) \\
  & + \delta x^\top \left( \frac{\partial M}{\partial t} + 
    D_f M + \sum_{i=1}^m u_i D_{b_i}M \right)\delta x ds \\
  = &  \int_0^1 \delta x^\top 
  \left( \frac{\partial M}{\partial t} + D_f M + 
  \he \left\{ M\frac{\partial f}{\partial x} \right\}
  \right) \delta x \\
   & + \sum_{i=1}^m u_i \delta x^\top \left( D_{b_i}M + \he \left\{
    M\frac{\partial b_i}{\partial x} \right\} \right) \delta x\\
    & + 2\delta x^\top \sum_{i=1}^m \frac{\partial u_i}{\partial s} Mb_i ds.
\end{align*}Invoking {\bf A2}, the fact that $\frac{\partial h}{\partial x} v_s =0$
and plugging in \eqref{eq:PDE}, we immediately get 
$
{d\over dt} \int_0^1 \delta x^\top M \delta x ds
=
\int_0^1 \delta x^\top 
  \left( \frac{\partial M}{\partial t} + D_f M + 
  \he \left\{ M\frac{\partial f}{\partial x} \right\}
   - \rho MBB^\top M
  \right) \delta x ds.
$

Define $\mathfrak{a} = \delta x^\top (\frac{\partial M}{\partial t} + D_f M + 
\he \left\{ M\frac{\partial f}{\partial x} \right\} - 2\lambda M)\delta $ and 
$\mathfrak{b} = \delta x^\top MBB^\top \delta x$. Invoking {\bf A1}, we set $\rho$ as $0$
if $\mathfrak{a}<0$ and $(\mathfrak{a}+\sqrt{\mathfrak{a}^2+\mathfrak{b}^2})/
\mathfrak{b}$ otherwise. It follows that 
\[
  \frac{d}{dt}\int_0^1 \delta x^\top M \delta x ds < 
  - 2 \lambda \frac{d}{dt}\int_0^1 \delta x^\top M \delta x ds,
\]
showing that $u(t,1)$ exponentially stabilizes
the trajectory $x_{\star}(\cdot)$, i.e., there exist two constants
$K, \lambda>0$ such that
\begin{equation} \label{converge: openloop}
d_{{M}}(x_{\star}(t),x(t))\le K e^{-\lambda(t-t_{0})}
d_{{M}}(x_{\star}(t_{0}),x(t_{0}))
\end{equation}where $d_M(\cdot, \cdot)$ stands for the Riemannian distance on the manifold $\calm$ under the metric $M(x,t)$. Note that here $K$ may be
larger than $1$ since the initial curve $\gamma$ is the geodesic under 
the induced metric.


2) \emph{Sampled-data controller.}
Choose a sampling time $T>0$ satisfying
$
K\sqrt{\frac{a_{2}}{a_{1}}}e^{-\lambda T} =: k <1,
$ 
in which the constant $K$ is the same as in \eqref{converge: openloop}.
At each sampling time instant $t_{i}$, $i = 1,2, \cdots$, 
measure the state $x(t_{i})$, compute a minimizing geodesic $\gamma_{i}$ connecting $x_{\star}(t_{i})$ to $x(t_{i})$ under the \emph{induced metric} from the ambient space and apply the open-loop control described
above with initial condition $x(t_{i},s)=\gamma_{i}(s)$
on the interval $[t_{i},t_{i+1})$. Then, such a sampled-data feedback exponentially stabilizes the trajectory $x_\star(\cdot)$.
Let $d(\cdot, \cdot)$ be the Riemannian distance corresponding to
the induced metric from the ambient space on $\calm$ and $d_M(\cdot, \cdot)$ the one corresponding to the metric $M$. Invoking \eqref{converge: openloop} and Lemma \ref{lem:estim-dist}, we have for all $i\in \mathbb{N}_+$ and $k \in (0,1)$,
\begin{align*}
    d(x(t_{i+1}), x_\star(t_{i+1}))
    & \le \frac{1}{\sqrt{a_1}} d_{{M}}(x(t_{i+1}), x_\star(t_{i+1})) \\
    & \le \frac{1}{\sqrt{a_1}} K e^{-\lambda T} 
     d_{{M}}(x(t_{i}), x_\star(t_{i})) \\
    & \le \sqrt{\frac{a_2}{a_1}} K  e^{-\lambda T} 
    d(x(t_i),x_\star(t_i)) \\
    & = k d(x(t_i),x_\star(t_i)) .
\end{align*}
It is then standard argument to show  
$$
d_{{M}}(x(t), x_\star(t)) \le K e^{-\tilde{\lambda}(t-t_0)} 
d_{{M}}(x(t_0), x_\star(t_0))
$$
for some $K', \tilde{\lambda}>0$ and all $t\ge t_0$.

It complets the proof.
\end{proof}

\begin{rem}
Note that it is unnecessary to choose the initial curves $\gamma(\cdot)$ and $\gamma_i(\cdot)$ as the geodesics -- under the metric $M(x,t)$ -- for the open-loop and sampled-data controllers. In fact, thanks to assumption {\bf A1}, there exist  constants $c_1, c_2>0$ such that $c_1 I \le {M}(x,t) \le c_2 I$ on $T \calm$.
Therefore, by Lemma \ref{lem:estim-dist} in Appendix, the geodesic distance 
under the metric ${M}$ is equivalent to the one induced by the metric of the ambient space. This idea can be also used to design a sampled-data feedback  when 
the geodesic on $\calm$ under the induced metric, e.g. Euclidean metric, from the ambient space can be easily computed.  
\end{rem}

\begin{rem}
    The conditions {\bf A0-A2} have two drawbacks from a numerical implementation perspective. Firstly, these conditions are non-convex, which require heavy computational burden. Second, when the dimension of the manifold $\calm$ is much smaller
    than that of the ambient space, numerous unnecessary computations will arise, adding an undue computational burden. More precisely, one needs to search for
    a matrix function taking values in $\mathbb{R}^{n \times n}$, while $\calm$ only has the dimension $q \ll n$.
\end{rem}

\subsection{Convexified conditions}

In this section, we propose a modified version of conditions to address the aforementioned numerical challenges. The results are particularly useful on systems evolved on Lie groups. Toward this end, we make the following assumption.

\begin{itemize}
    \item[\bf A3] The tangent bundle of $\calm$ in \eqref{eq:M} is spanned by \emph{independent} vector fields $\{s_1, \cdots, s_{q} \}$. Assume $\|S(S^\top S)^{-1}\|\le c_S, ~\forall x\in \calm$ with some $c_S$ and the definition $S(x) := [s_1(x), \cdots, s_q(x)]$.
\end{itemize}

Under the above assumption, we are able to find some smooth matrix function 
$E: \calm \times \rea_{\ge 0} \to \mathbb{R}^{q \times m}$  such that 
\begin{equation} \label{B-KE}
    B(x,t) = S(x) E(x,t),
\end{equation}
due to the fact $b_i \in T \calm$. Likewise, $\delta x \in T\calm $ can be written as $\delta x = S(x) v(x) $ for some $v(x) \in \mathbb{R}^q $. Now, instead of searching for $M \in \mathbb{R}^{n\times n}$ directly, we may search for another metric $\mathfrak{M}(x,t)$ in a lower-dimensional space $\rea^{q\times q}$, and the previous metric $M$ is parameterized by 
\begin{equation} \label{P_S}
    {M}(x,t) = P_S(x) \mathfrak{M}(x,t) P_S^\top (x)
\end{equation}
where $P_S = S(S^\top S)^{-1}$ is the projection operator which is computed before searching for the CCM.

Considering the boundedness of $P_S$, the condition {\bf A0} is equivalent to:

\begin{itemize}
    \item[{\bf A0$'$}] there exist two positive constants $a_{1},a_{2}$ such that
\end{itemize}
\begin{equation} \label{A0'}
    a_1 I_q  \le \frm(x,t) \le a_2 I_q, \quad 
    \forall x \in \calm, \; t\ge 0.
\end{equation}
Meanwhile, writing $\delta x = S(x) v(x) \in T \calm$,
and invoking \eqref{B-KE} and \eqref{P_S}, we have 
$
B^\top M \delta x = E^\top \frm v.
$
Then, {\bf A1}-{\bf A2} can be reformulated as:
\begin{itemize}
    \item[{\bf A1$'$}] For $ v \in \mathbb{R}^q \backslash \{0 \}$,
the following implication holds
\begin{equation}
\begin{aligned}
E^\top &\frm v =0 \implies
\\
 & v^{\top}\left(\frac{\partial \frm}{\partial t}  
  +D_{f}\frm 
  + \he\{\frm S_f\} +2\lambda \frm\right)v<0 
\end{aligned}\label{A1'}
\end{equation}
where the known matrices $S_f  := D_f (P_S^\top) S + P_S^\top \frac{\partial f}{\partial x} S$ and $S_{b_i}  := D_{b_i} (P_S^\top) S + P_S^\top \frac{\partial b_i}{\partial x} S$ are in $\rea^{q\times q}$.

    \item[{\bf A2$'$}] For each $i\in \ell_m$:
\end{itemize}
\begin{equation} \label{A2'}
    D_{b_i}{\frm}  + S_{b_i}^\top {\frm} + {\frm} S_{b_i} =0.
\end{equation}

Following \cite{manchester2017control}, we convexify {\bf A0$'$} - {\bf A2$'$} 
via the ``musical isomorphism'' $W^{-1}(x,t) = {\frm}(x,t)$. Then, \eqref{A0'} and \eqref{A2'} are equivalent to
\begin{equation} 
\label{W:A0}
    \frac{1}{a_2} I_q \le W \le \frac{1}{a_1} I_q,
\end{equation}
and\footnote{Due to topological obstructions, it is sometimes 
challenging to find a metric under which ($b_i$)s are Killing. In
that case, one needs to consider the weak form of CCM.}
\begin{equation} 
\label{W:A2}
  -D_{b_i} W + WS_{b_i}^\top + S_{b_i} W =0.
\end{equation}
respectively. 

For \eqref{A1'}, it is equivalent to the 
existence of a scalar function $\rho(x,t)$ such that for all $x,t$:
\begin{equation} \label{W:A1}
         -\frac{\partial W}{\partial t}
           - D_f {W}  
           + W S_f^\top  + S_f {W} 
           + 2 \lambda {W} 
           - \rho EE^\top < 0.
\end{equation}
The formulas \eqref{W:A0} - \eqref{W:A1} is convex w.r.t. to the metric $W$ to be searched for. Once $\rho$ and $W$ has been obtained, we can proceed to controller design with the CCM 
\begin{equation} \label{tildeM}
    M = P_S W^{-1} P_S^\top.
\end{equation}

\begin{rem}
In Lie groups, the above condition {\bf A3} is always  guaranteed as the tangent bundle of a Lie group $G$ is trivial in the sense that $TG = G \times \mathfrak{g}$, see e.g., \cite{lee2012smooth}.
Thus, for example, $TG$ is spanned by left (right)-invariant vector fields. On the other hand, generally even if there exist no global vector fields spanning $T\calm$, it is always possible to work locally.
\end{rem}

\subsection{CCM on Lie groups}
Most of Lie groups in control applications can be modeled as $G \times \mathbb{R}^l$, where $G$  is a Lie subgroup of $GL(\mu) = \{A \in \mathbb{R}^{\mu\times \mu}: 
A \text{ invertible}\}$.\footnote{More generally, one can consider $G \times (\mathbb{S}^1)^p \times \rea^l$. Then, the Lie group can be embedded in $\rea^{\mu^2 + 2p + l}$.} This class of Lie groups is naturally embedded in the Euclidean space $\rea^{\mu^2+ l}$.

We use the space $O(2) \times \mathbb{R}$ as a case study to illustrate how to apply the methods proposed in the previous subsection.

{\it Step 1}: We embed the Lie group $O(2) \times \mathbb{R}$ into $\rea^{5}$ via
\[
 O(2)\times \rea \ni (R, x) \mapsto (r,x):=[R_{11}, R_{12}, R_{21}, R_{22}, x]^\top
\]
where $R_{ij}$ are the elements of $R\in O(2)$. The constraint map $h$, with $q=3$, is then
\[
    h(r,x) = 
    \begin{bmatrix}
        r_1^2 + r_3^2 - 1 \\
        r_2^2 + r_4^2 -1  \\
        r_1 r_2 + r_3 r_4  
    \end{bmatrix},
\]
which is a reinterpretation of $R^\top R = I_2$. From $h$, we
can determine the matrix
$S(r,x)$. This is rather straighforward:
\[
S(r,x) = 
\begin{bmatrix}
    0 & 0 & 0 & 0 & 1\\
    -{r_2 \over \sqrt{2}}   & {r_1  \over \sqrt{2}}  & -{r_4 \over \sqrt{2}} & {r_3  \over \sqrt{2}} & 0
\end{bmatrix}^\top.
\] 
In addition, we have $S^\top S = I_2$, and thus {\bf A0} is satisfied. It yields  $P_S (r,x) = S(r,x)$. 

{\it Step 2:}
Write the system dynamics into standard form \eqref{sys:nltv}
and search for
$W\in \rea^{2 \times 2} $ satisfying \eqref{W:A0}-\eqref{W:A1}.

{\it Step 3:}
Solve the partial differential equation \eqref{eq:PDE} -- the path integral -- to obtain the controller. We need to calculate the geodesic on $O(2)\times \mathbb{R}$
under the induced metric from the ambient space.
Since $O(2)$ is compact, the geodesic is nothing but given by the matrix
exponential. 
That is, given $R_1, R_2 $ in the same component of $O(2)$,
the geodesic joining $R_1$ to $R_2$ is given by 
$t \mapsto R_1 \exp(\log(R_1^\top R_2)t)$. Thus, the geodesic on $O(2)\times \mathbb{R}$
joining $(R_1, x_1)$ to $(R_2, x_2)$
is $t \mapsto (R_1 \exp(\log(R_1^\top R_2)t), (1-t)x_1 + t x_2) $.

Following the above three steps, we are able to design CCM-based controllers for general control systems evolving on $G \times G_1 \times \mathbb{R}^l$, where $G$ is a compact Lie group in $GL(\mu)$ and $G_1$ represents a Lie group embedded in $\mathbb{R}^p$.

\begin{rem}
In {Step 1}, it is unnecessary to compute $S$ from 
the constraint $h(r,x)$. As we mentioned before, the tangent bundle
of a Lie group is trivial and can be simply chosen as the left (or 
right) invariant vector fields, which is isomorphic to the Lie algebra.
For example, the Lie algebra of $O(2)\times \rea$ is 
$\mathfrak{o}(2) \times \rea$ for a basis is
\[
    E_1 = 
    \begin{bmatrix}
        0 & 0 \\ 0 & 0 
    \end{bmatrix} \times \{ 1\}.
    \quad
    E_2 = 
    \begin{bmatrix}
        0 & {1\over \sqrt{2}} \\ {-1\over \sqrt{2}} & 0
    \end{bmatrix} \times \{ 0 \},
\]
The left-invariant vector fields corresponding to $E_1,E_2$ are
\[
\begin{aligned}
    S_1  = 
    \begin{bmatrix}
        0 & 0 \\ 0 & 0 
    \end{bmatrix} \times \{ 1 \} , \;
    S_2  = R
    \begin{bmatrix}
        0 & {-1\over \sqrt{2}} \\ {1\over \sqrt{2}} & 0
    \end{bmatrix} \times \{ 0 \},
\end{aligned}
\]from which one easily recovers $S(r,x)$.

\end{rem}

{\it Example 1:} Consider the space 
$SE(3)$:
$$ 
SE(3)=\left\{ \begin{bmatrix}R & v\\
0 & 1
\end{bmatrix}
\in
\mathbb{R}^{4\times4}:R^{\top}R=I,\det R=1,v\in\mathbb{R}^{3}\right\} .
$$
which is isomorphic to $SO(3) \times \mathbb{R}^3$. Embed $SE(3)$ into $\rea^{12}$ as in the previous subsection and denote the state variable in $\rea^{12}$ as $(r,x)$.
Note that $\dim \{SE(3)\} = 6$, we shall search for a CCM $M(x,t)\in \rea^{6\times 6}$.
Since $\mathfrak{se}(3) = \mathfrak{so}(3)\times \rea^3$, it is
easy to calculate $S(r,v) \in \rea^{12 \times 6}$ as
\[
    S(r,x) = 
    \begin{bmatrix}
        0_{9\times 3} & S_r \\ I_3 & 0_{3\times 3}
    \end{bmatrix}
\]in which
\[
    S_r = \frac{1}{\sqrt{2}}
    \begin{bmatrix}
        -r_2 & -r_3 & 0 \\
        r_1 & 0 & -r_3 \\
        0 & r_1 & r_2 \\
        -r_5 & -r_6 & 0 \\
        r_4 & 0 & -r_6 \\
        0 & r_4 & r_5 \\
        -r_8 & -r_9 & 0 \\
        r_7 & 0 & -r_9 \\
        0 & r_7 & r_8
    \end{bmatrix}
\]
One can verify that $S^\top S = I_6$. As a consequence, $P_S(r,x) = S(r,x)$.

Consider the system dynamics
\begin{align*}
\dot{R} & =R\Omega\\
\dot{v} & =-kv+Re,
\end{align*}
in which $e$ is a fixed unit vector in $\mathbb{R}^{3}$
and $\Omega\in\mathfrak{so}(3)$ can be directly controlled.
The control objective is to make $v(t)$ exponentially converge to a given desired feasible $ v_{\star}(t)$.

The matrices $S_f$ and $S_{b_i}$ can be easily computed:
\[
S_f = 
\begin{bmatrix}
    -k I_3 & (I_3 \otimes e^\top) S_r \\
    0_{3\times 3} & 0_{3\times 3}
\end{bmatrix}, \quad
S_{b_i} =
\begin{bmatrix}
   0_{3\times 3} & 0_{3\times 3} \\ 0_{3\times 3} & F_i 
\end{bmatrix}
\]where
\begin{align*}
F_1 &  = 
\frac{1}{2}
\begin{bmatrix}
   r_2-r_4 & r_3 & r_6  \\
   0 & -r_4 & -r_5 \\
   0 & r_1 & r_2
\end{bmatrix}, \\
F_2 & = 
\frac{1}{2}
\begin{bmatrix}
    -r_7 & 0 & r_9 \\
    r_2 & r_3 - r_7 & -r_8 \\
    -r_1 & 0 & r_3
\end{bmatrix}, \\
F_3 & =
\frac{1}{2}
\begin{bmatrix}
    -r_8 & -r_9 & 0 \\
    r_5 & r_6 & 0 \\
    -r_4 & -r_7 & r_6 - r_8
\end{bmatrix}.
\end{align*}

We can then substitute these matrices into \eqref{W:A0}-\eqref{W:A1}
to solve for the matrix function $W$.





\subsection{More abstract manifolds}

In this subsection, we provide an intrinsic treatment of CCM on abstract manifolds whose embeddings into Euclidean spaces are not immediately obvious. These results are of theoretical interest and may be viewed as a high level guideline -- thanks to their much
simpler forms. 

Consider a Riemannian manifold $\calm$ equipped with a metric $g_0$.
Let $x(t,s)$ be the solution to the system \eqref{sys:nltv} with initial
condition $\gamma(s)$ at $t=0$, with $\gamma: [0,1]\to \calm$ 
a geodesic under the metric $g_0$.

We shall look for a CCM ${g}$, such that
\begin{equation} \label{Riem: CCM}
    \frac{1}{2} \frac{d}{dt}
    g\left(
    \frac{\partial x(t,s)}{\partial s}, \frac{\partial x(t,s)}{\partial s}
    \right)
    \le - \lambda 
    g\left(
    \frac{\partial x(t,s)}{\partial s}, \frac{\partial x(t,s)}{\partial s}
    \right)
\end{equation} for some $\lambda >0$ and all $u, \, t\ge 0, \, s\in [0,1]$.
Let $D \over ds$ be the covariant derivative associated with the metric $g$.
For notational ease, denote $v_s:= \frac{\partial x(t,s)}{\partial s}$,
then the left hand side of \eqref{Riem: CCM} can be calculated as
\begin{align*}
     & \frac{1}{2} \frac{dg\left(v_s, v_s \right) }{dt}   \\
     =&    g\left(\frac{D v_s}{dt}, v_s\right) + \dot{g}( v_s, v_s) \\
     =&  g\left(\nabla_{v_s} \left(f 
       +\sum_{i=1}^m u_i b_i\right), v_s \right)+ \dot{g}( v_s, v_s)\\
     =& g(\nabla_{v_s}f + \sum_{i=1}^m u_i \nabla_{v_s} b_i 
         + \sum_{i=1}^m \frac{\partial u_i}{\partial s} b_i , v_s)
         + \dot{g}( v_s, v_s)\\
     =& g(\nabla_{v_s}f , v_s) + \sum_{i=1}^m u_i g(\nabla_{v_s} b_i,v_s) 
       + \sum_{i=1}^m \frac{\partial u_i}{\partial s} g(b_i , v_s)  \\
      & \quad \hspace{4.1cm} \quad + \dot{g}( v_s, v_s). 
\end{align*}
If \eqref{Riem: CCM} holds for all $u$,
the second term on the last line vanish since it depends linearly
on $u$. In other words,
\begin{equation}
    g(\nabla_{v} b_i, v) =0, \quad \forall v \in T\calm, \; i \in \ell_m.
\end{equation} This is nothing but saying that $b_i$s are Killing fields (c.f. 
{\bf C1}).

To continue, note that if for $v\in T\calm$ satisfying
$g(b_i, v) =0 $ for all $i \in \ell_m$, there holds
\begin{equation} \label{Riem: CCM1}
    g(\nabla_{v} f, v) + \dot{g}(v, v)  + 2 \lambda g(v, v)<0,
\end{equation}
then we can simply design a ``differential controller'' as
$\frac{\partial u_i}{\partial s} = -\rho g(b_i, v_s)$
for some non-negative function $\rho(x,t)$ as
in Proposition \ref{prop:1}. We underline that \eqref{Riem: CCM1}
is exactly \eqref{A1}. 

Like in the submanifold case, 
the search for a CCM on an abstract manifold can be convexified via
the musical isomorphism: $\flat: T\calm \to T\calm^*$. In words, given
$v\in T\calm$ and the Riemannian metric $g$, we lift $v$ and $g$ to $T\calm^*$:
\begin{align*}
     v&=v^i\frac{\partial}{\partial x_i} 
      \mapsto  v^i g_{ij} dx^j \\
     g& = g_{ij}dx^i dx^j \mapsto g^{ij} \frac{\partial}{\partial x_i} \frac{\partial}{\partial x_j}
\end{align*}where $(g^{ij})$ is the inverse of $(g_{ij})$.

For implementation, all the expressions will need to be written in local
coordinates, from where one can recover the assumptions {\bf A0}-{\bf A2}. In
general, however, one should not expect to solve the problem globally
like on Lie groups.

\section{Concluding remarks}

This paper has presented an extension of the control contraction metrics approach from Euclidean space to Lie groups, in which we view the manifolds as constrained sets.
In particular, we show that the search for CCM on matrix Lie groups (potentially
with a direct product with a vector space) can be formulated as the convex conditions. Future directions would be to apply the proposed approach to study the trajectory tracking control of some  practical systems on manifolds.

\appendix


\begin{lem} \label{lem:estim-dist} 
Let $g_1, g_2$ be two Riemannian metrics on a manifold $\calm $ satisfying
\begin{equation}
a_{1}g_1(v,v)\le g_2(v,v)\le a_{2}g_1(v,v),\ \ \forall v\in T\mathcal{M}
\end{equation}for some positive constants $a_1, a_2$.
Let $d_1$ and $d_2$ be the induced Riemannian distances of the two metrics
respectively. Then there holds
\begin{equation}
\sqrt{a_{1}}d_{1}(x,y)\le d_{2}(x,y)\le\sqrt{a_{2}}d_{1}(x,y),\ \ \forall x,y\in\mathcal{M}\label{dist-bd}
\end{equation}
where $d_g$ and $d_{\bar g}$ are the induced distances on $\cal M$, whenever they are well defined.
\end{lem}
\begin{proof}
For a given pair of points $x,y\in\mathcal{M}$ such that $d_{g}(x,y)$
and $d_{\bar{g}}(x,y)$ are defined, let $\gamma_{g}:[0,1]\to\mathcal{M}$
and $\gamma_{\bar{g}}:[0,1]\to\mathcal{M}$ be the minimizing geodesics
joining $x$ to $y$ for $d_{g}$ and $d_{\bar{g}}$, respectively.

From Assumption {\bf A1}, we have on one hand,
\begin{align*}
d_{\bar{g}}(x,y) & =\int_{0}^{1}\sqrt{\bar{g}(\gamma_{\bar{g}}'(s),\gamma_{\bar{g}}'(s))}ds\\
 & \le\int_{0}^{1}\sqrt{\bar{g}(\gamma_{g}'(s),\gamma_{g}'(s))}ds\\
 & \le\sqrt{a_{2}}\int_{0}^{1}\sqrt{g(\gamma_{g}'(s),\gamma_{g}'(s))}ds\\
 & =\sqrt{a_{2}}d_{g}(x,y),
\end{align*}
with the notation $(\cdot)' = \frac{\partial (\cdot)}{\partial s}$ for scalar functions,
and on the other hand, similarly
\begin{align*}
d_{\bar{g}}(x,y) & =\int_{0}^{1}\sqrt{\bar{g}(\gamma_{\bar{g}}'(s),\gamma_{\bar{g}}'(s))}ds
 \ge\sqrt{a_{1}}d_{g}(x,y)
\end{align*}
This completes the proof.
\end{proof}

\bibliographystyle{IEEEtran}
\bibliography{IEEEabrv,GeometricControl}

\end{document}